\documentclass[a4paper,fleqn]{article}
\usepackage{amsmath}
\usepackage{latexsym}
\usepackage{amssymb}
\usepackage{amsfonts}
\usepackage{rsfs}
\usepackage[dvips]{graphics}
\usepackage[all]{xy}
\usepackage{amsthm}
\usepackage{enumerate}
\usepackage{url}
\DeclareMathOperator{\dotlr}{\overset{\centerdot}{\leftrightarrow}}

\DeclareMathOperator{\gal}{\text{gal}}

\DeclareMathOperator{\Set}{\mathit{Set}}
\DeclareMathOperator{\Sms}{\mathit{Sms}}

\DeclareMathOperator{\Fin}{\mathit{Fin}}

\newcommand*{\FN}{\mathit{FN}}

\let\uhr\upharpoonright
\renewcommand*{\upharpoonright}{\hspace{-.07cm}\uhr\hspace{-.07cm}}
\newtheorem{thm}{Theorem}
\newtheorem{prop}{Proposition}
\newtheorem{cor}{Corollary}
\newtheorem{lemma}{Lemma}
\theoremstyle{remark}

\theoremstyle{definition}
\newtheorem{defn}{Definition}
\newtheoremstyle{axiom1}
{3pt}
{3pt}
{\itshape}
{}
{\bfseries\itshape}
{}
{.5em}
{}
\newtheoremstyle{axiom2}
{3pt}
{3pt}
{\itshape}
{}
{\bfseries\itshape}
{\\}
{.5em}
{}
\theoremstyle{axiom1}
\title{An Alternative Set Model of Cognitive Jump}
\author{Kiri Sakahara\thanks{Yokohama National Universuty and Kanagawa University, Kanagawa, Japan.} \and Takashi Sato\thanks{Toyo University, Tokyo, Japan.}}
\date{}
\begin{document}

\maketitle

\begin{abstract}
When we enumerate numbers up to some specific value, or, even if we do not specify the number, we know at the same time that there are much greater numbers which should be reachable by the same enumeration, but indeed we also congnize them without practical enumeration. Namely, if we deem enumeration to be a way of reaching a number without any ``jump'', there is a ``jump'' in our way of cognition of such greater numbers.

In this article, making use of a set theoretical framework by Vop\v{e}nka \cite{ast} (alternative set theory) which describes such structure, we attempt to shed light on an analogous sturucture in human and social phenomenon. As an example, we examine a problem of common knowledge in electronic mail game presented by Rubinstein \cite{rubinstein89}. We show an event comes to common knowledge by a ``cognitive jump''. 
\end{abstract}

\section{Introduction}
Asked about how many sands are there in a shore, what would the answer be?
Some may figure out a rough estimate effortlessly.
But most of us would not be able to give even a single guess and feel that it is too many to count, almost as infinite.

Every time we try to count the very huge pile of objects, this kind of feeling would be revived. 
The word \textit{infinity} may not express any estimates nor specific numbers but do this impossibility itself of achieving the task within any given time.
It may be possible to count them all somehow, but, in most cases, we helplessly jump to the conclusion that there is infinite grain of sands.

It is this kind of impossibility which marks not only our own intuitive understanding of it but the way mathematics construct infinity in its system.
Infinite sets are understood to be the ones which cannot be reached by any inductive operations. 
To be more precise, it is stated that there exists an \textit{inductive set} which includes all successors of its elements besides empty set, so that the set cannot be successors of any elements of it and, thus, is closed under induction.

While this construction opens the way to deal with infinite calculus in set environment, 
some problematic sides also surface gradually when they are applied to model some real-life situations.
The one we focus here is the problem how it is possible that some informations to become common knowledge.

This problem has been considered and dealt within a standard set environment, or Zermelo-Fraenkel set theory.
However, it seems that the problem stems not from the way we handle the problem, but from the very way a standard set theory deals with infinity itself.
In fact, by reconsidering the way we handle infinity from an alternative view, it is possible to solve the problem very smooth and natural way as we shall see.

In the present paper, we examine one of the problematic sides of how a standard set theory deals with infinity along with the problem of common knowledge building and offer a new way to handle infinite phenomenons.

\section{Preliminaries}
Our framework is built in a universe of alternative set theory (AST), which is originated and developed by Vop\v{e}nka \cite{ast} to construct a set theoretical environment in which one can deal with the phenomena involving infinitely many repetitions in accordance with one's own intuitive understanding, or phenomenological view \cite{ast}, of them.
Having unique and attractive structure,
it is almost forgotten, so we first overview its characteristic features and corresponding axioms before describing our framework\footnote{We consult Vop\v{e}nka and Trlifajov\'{a} \cite{encycro-ast} as a reference here.}.

AST differs from Zermelo-Fraenkel set theory (ZF) mainly in
the way it treats infinity. 
While ZF treats infinity as, say, \textit{actual infinity}, AST deems it as \textit{subjective infinity}, or \textit{natural infinity} \cite{vop2} in his own words.
This difference is reflected in the way it constructs hierarchy between sets and classes.

ZF postulates an axiom of infinity, which assumes the existence of inductive set, which contains all successors of its elements, while AST does not.
Instead, AST requires existence of infinite proper classes, so-called \textit{semisets}, and infinite sets are defined as the sets which include semisets as its subclasses.

This unique structure permits AST to model non-standard framework of arithmetic, which include, for example, \textit{huge} natural numbers, in an intuitively sound way.
Huge natural numbers are defined as the number which is bigger than any finite natural number.
However, the entire collection of finite natural number cannot be, or at least seems not to be, described explicitly, so that the class of \textit{all} finite natural numbers cannot form a set, it remains only a class.
Consequently, every huge number, which itself is a set, must include a proper class as its subclass, or, a semiset.

In fact, the predicate \textit{huge} can properly be used only when some subjects are there to use it in a proper manner.
This subjective aspect of hugeness can only be properly grasped when we regard infinity as some sort of subjective notion.
Our framework depends on this structure: the existence of the huge numbers, or, infinity in our daily life.
Taking into account of this structure, we cast a new light on how to interpret our own judgement made in everyday life, and propose a new way to model these judgement.

We start with the axiom of existence of sets.
For notational ease, we use a notation $\Set(x)$ to mean ``$x$ is a set".
\newtheorem*{ax03}{Axiom of existence of sets (empty set and set-successors)}
\begin{ax03}
\[
\Set(\emptyset)\wedge\left(\forall x\right)\left(\forall y\right)\Set(x\cup\{y\})
\]
\end{ax03}

To introduce further axioms we formally define \textit{set-formulas} as usual.
Set-formulas are built up from two types of set-formulas $x = y$ and $x\in y$, where $x$ and $y$ are sets, by means of five connectives $\left(\varphi \wedge \psi \right),\ \left(\varphi\vee \psi \right),\ \left(\varphi\Rightarrow \psi \right),\ 
\left(\varphi\equiv\psi \right)$, $\neg\left(\varphi\right)$ and two quantifiers $\left(\exists x\right)\varphi,\ \left(\forall x\right)\varphi$, where $\varphi$ and $\psi$ are set-formulas.

\theoremstyle{axiom2}
\newtheorem*{ax04}{Axiom of induction}
\begin{ax04} 
Let $\varphi\left(x\right)$ be a set-formula. 
Then 
\[
\varphi\left(\emptyset\right) \wedge \left(\forall y\right)
\bigl(\varphi\left(x\right)\Rightarrow\varphi\left(
x\cup\{y\}\right)\bigr)\ \Rightarrow\ \left(\forall x\right)\varphi\left(x\right)
\]
\end{ax04}

Sets in AST are generated inductively by these axioms in a concrete way.
Alongside with these axioms, AST also introduces extensionality for sets and regularity, while pairing, union, power set, schema of separation and replacement axioms are omitted, since these axioms can be derived from these axioms.
Notice that the axioms for sets coincide with that of ZF's except that the axiom of infinity is negated as we saw above.

In contrast, classes are only postulated.
\newtheorem*{ax07}{Axiom of existence of classes}
\begin{ax07}
For each property $\varphi$, there exists the class $\{x\,;\, \varphi(x)\}$
\end{ax07}

If $\varphi(x)$ is a set-formula then we say that the class $\{x\,;\, \varphi(x)\}$ is \textit{set-theoretically definable}.
For classes, the axiom of extensionality is also stated.
We denote classes in upper case $X, Y,...$ and sets in lower case $x, y,...$.

Semisets are defined as subclasses of sets.
The existence of the special kind of semisets is guaranteed by the next axiom.
The notation $\Sms(X)$ means ``$X$ is a semiset."
\theoremstyle{axiom2}
\newtheorem*{ax10}{Axiom of existence of proper semisets}
\begin{ax10}
There is a proper semiset. In Symbols,
\[
\left(\exists X\right)\left(\Sms\left(X\right) \wedge \neg\Set\left(X\right)\right).
\]
\end{ax10}

The existence of proper semisets implicitly premises some largeness of sets which include proper semisets as their subclasses.
If sets are small enough to look inside and check their contents, they can be identified, and every subclasses must be sets.
On the other hand, when a set consists of infinitely many contents, exactly how many elements are included remains indeterminate, so that some subclasses cannot be sets: they must be semisets.

The finiteness and infinity in AST are defined in accordance with this intuition.
A class $X$ is \textit{finite} (notation: $\Fin(X)$) iff each subclass of $X$ is a set.
On the other hand, a class is \textit{infinite} iff it includes proper semisets.
Additionally, a class $X$ is called \textit{countable} iff $X$ is infinite class with a linear ordering $\leq$ such that each segment $\{y\in X\,|\, y\leq x\}$ is finite.

\newtheorem*{ax11}{The prolongation axiom}
\begin{ax11}
For each countable function $F$, there is a set function $f$ such that $F\subseteq f$.
\end{ax11}

This axiom leads to existence of infinite sets.
Since every function whose domain is proper semiset can be prolonged to the function whose domain is the set which include the semiset as its subclass, the domain of this function forms an infinite set.
But cardinality of infinite sets is differentiated from that of semisets.
This premise is stated in the next axiom.
To state the axiom, we introduce a relation between two classes.
Two classes $X$, $Y$ are \textit{equivalent} iff there is a one-one mapping $X$ onto $Y$, i.e. $X\approx Y$.

\newtheorem*{ax13}{Axiom of cardinalities}
\begin{ax13}
For all infinite classes $X$, $Y$ if both $X$ and $Y$ are not countable then they are equivalent.
This cardinality is called uncountable.
\end{ax13}

Evidently from the axiom, any infinite class is countable or uncountable.
In addition to them, the axiom of choice is also adopted.

In this environment of AST, a model of Peano arithmetic can be built.
The class of \textit{natural numbers} $N$ is defined, as usual, in the von Neumann way.
\[
N\ =\ \left\{x\ ;\ 
\begin{matrix}
\left(\forall y\in x\right) \left(y\subseteq x\right)\\
\wedge\left(\forall y,z\in x \right) \left(y\in z \vee y=z \vee z\in y\right)
\end{matrix}
\right\}
\]

The class of \textit{finite natural numbers} $\FN$ consists of the numbers represented by finite sets.
\[
\FN \ =\ \left\{x\in N\ ;\ 
\Fin(x)\right\}
\]
We also name $N\setminus\FN$ as the class of \textit{huge natural numbers}.

Apparently from the definition, $\FN$ is countable, and each huge natural number, which include $\FN$ as its semiset, is an uncountable set. 
It is also evident that $N$ is an uncountable proper class.
Since $\FN$ is a semiset, $N\setminus\FN$ has no least element with respect to $\subseteq$ such as $\omega$, the least transfinite ordinal number of ZF.

Lastly, we introduce a structure having essentially the same feature of semisets:
a $\sigma$-equivalence \cite{gurican}.
$\sigma$-equivalence is defined as the equivalence generated from the union of countably many set-theoretically definable classes.
It is introduced originally, in Guri\v{c}an and Zlato\v{s} \cite{gurican}, as topological structure to decide whether any given pair of sets are accessibile or not, and named as \textit{accessibility equivalence}.
We use this equivalence later to assess accessibility between two states.

\begin{defn}
A class $X$ is a $\sigma$-class if $X$ is the union of a countable sequence of set-theoretically definable classes.
A class $\dotlr$ is a \textit{$\sigma$-equivalence} if $\dotlr$ is a $\sigma$-class and an equivalence relation.
A sequence $\left(R_n\right)_{n\in \FN}$ is a \textit{generating sequence} of an equivalence $\dotlr$ iff the following conditions hold:
\begin{enumerate}[(1)]
\item For each $n$, $R_n$ is a set-theoretically definable, reflexive, and symmetric relation.
\item For each $n$ and each $x$, $y$, $z$, $\langle x,y\rangle\in R_{n}$ and $\langle y,z\rangle\in R_{n}$ implies $\langle x,z\rangle\in R_{n+1}$ ($R_n\circ R_n\subseteq R_{n+1}$) ;$\,R_0=\{\langle x,x\rangle; \text{ for all } x\}$.
\item $\dotlr$ is the union of all the classes $R_n$.
\end{enumerate}
\end{defn}

Even when each relation $R_n$ is not an equivalence, the union of them becomes an equivalence (since for every couple of pairs $\langle x,y\rangle\in R_n$ and $\langle y,z\rangle\in R_n$ there always exists $\langle x,z\rangle\in R_{n+1}$, thus transitivity is satisfied).
An equivalence $\dotlr$ is called \textit{totally disconnected} iff it is generated by an equivalence relation $S_n$ for each $n$.

Given a $\sigma$-equivalence $\dotlr$, the equivalent class of any given set $x$ can be formed, which is called a \textit{galaxy} of a set $x$ denoted as $\gal(x)\equiv\{y\in X\,;\,x\dotlr y\}$.
Almost immediately from the definition, galaxies are $\sigma$-class.
 
When the equivalence $\dotlr$ is totally disconnected, all galaxies are separated; there are no continuum line, or $\dotlr$-chain, 
between two different galaxies.
On the other hand, if it is not, there may be a $\dotlr$-chain starting from any given set $x$ and reaching out to the other set $y$ which is not in the galaxy $\gal(x)$ anymore.

It is noteworthy that there appears to be one of the characteristic, and little bit strange, structures of AST present here.
Provided that transitivity is valid, every element contained in one galaxy cannot go outside of it no matter how far it traces along any sequence of any chain under the normal ZF settings since it is equivalent to the first one.
But in AST, it is possible that the first one is not equivalent anymore after it traces huge steps.
This is the structure we focus on in the present paper.
We will explore its characteristics in the next section.

\section{Sorites Relation}\label{sorites}
To highlight the feature of the accessibility equivalence more concrete way, we make use of \textit{sorites relation}, mentioned in Tsujishita \cite{tjst}.
Let $R$ be a binary relation.
A sequence $(a_1,\ldots ,a_\alpha)$ where $\alpha\in N\setminus\FN$ is called an \textit{$R$-chain} if $a_iRa_{i+1}$ for $i\in\alpha$ (and $a_i\ne a_j$ for all $i\ne j$).

\begin{defn}
$R$ is called a \textit{sorites relation} iff $R$ is equivalence relation and there is an $R$-chain $(a_1,\ldots,a_\alpha)$ that is $\neg(a_1 R\, a_\alpha)$ for some $\alpha\in N\setminus\FN$.
\end{defn}

Then, our claim is stated as the next theorem.

\begin{thm}
There is a $\sigma$-equivalence $\dotlr$ that is a sorites relation.
\end{thm}

\begin{proof}
Let us start with defining a distance between two elements $a_\ell$ and $a_m$ of a sequence $\left(a_i\right)_{i\in N}$ as $||\alpha_\ell,\alpha_m||=|m-\ell|$.
We define a relation $R_n$ 
for each $n\in \FN$ as 
\[
R_n\ \equiv\ \{\langle a_i,a_j\rangle\,;\, ||a_i,a_j||<2^{n}\}.
\]
It is obvious that $\bigcup_{ n\in \FN} R_n$ is a $\sigma$-class and $\{R_n\,;\, n\in \FN\}$ is a generating sequence of an equivalence $\dotlr$, since for every $i,j$, $||a_i,a_j||=||a_j,a_i||$, thus every $R_n$ satisfies symmetry. Since $||a_i,a_i||=0\in \FN$, it also satisfies reflexivity.
Let $\langle x,y\rangle\in R_{n}$ and $\langle y,z\rangle\in R_{n}$.
By definition of $R_{n}$, $||x-y||<2^{n}$ and $||y-z||<2^{n}$, thus, $||x-z||\leq ||x-y||+||y-z||<2\times(2^{n})=2^{n+1}$.
It means $\langle x,z\rangle \in R_{n+1}$.

Let us next confirm the union of all the classes $R_n$ constructed above is a sorites relation.
Let $\alpha$ be a huge natural number which satisfies $\alpha=2^\gamma$ for some $\gamma\in N\setminus\FN$.
For any $i\in\alpha$, $a_i\dotlr a_{i+1}$ by definition.
But $||a_1,a_\alpha||=2^{\gamma}-1$, thus $\langle a_1,a_\alpha\rangle\notin R_i$ for all $i\in\FN$,
and $\neg(a_1\dotlr a_\alpha)$.
\end{proof}

Let us overview two basic properties of sorites relations.
Throughout the rest of the paper, all accessibility equivalences $\dotlr$ are supposed to be sorites relations unless they are specified.

\begin{prop}\label{basic1}
If $a_1\dotlr a_i$ then $a_1\dotlr a_{i+1}$.
\end{prop}
\begin{proof}
Since $(a_1,\ldots,a_\alpha)$ is $\dotlr$-chain, $a_i\dotlr a_{i+1}$ for all $i\in\alpha$.
$a_1\dotlr a_{i+1}$ follows by transitivity.
\end{proof}

\begin{prop}\label{basic2}
If $\neg\left(a_1\dotlr a_i\right)$ then $\neg\left(a_1\dotlr a_{i-1}\right)$.
\end{prop}
\begin{proof}
Suppose $a_1\dotlr a_{i-1}$, then $a_1\dotlr a_i$ follows since $a_{i-1}\dotlr a_i$.
It contradicts the assumption.
\end{proof}

As a direct result of these two propositions, a peculiar property of galaxies is drawn: for any given sorites sequence $(a_1,...,a_\alpha)$ and within any finite number of steps, $a_n$ cannot go outside of $\gal(a_1)$, or, conversely, penetrate into $\gal(a_\alpha)$.

The more peculiar property, which we mentioned in previous section, results when we ask exactly how many steps it takes to go across the border of a galaxy.
The answer cannot be specified.
Provided conversely that it takes, say $\beta<\alpha$ steps, then $a_\beta\in\gal(a_1)$ but $a_{\beta+1}\notin \gal(a_1)$.
It contradicts the Proposition \ref{basic1}.

It is strange. But it is the way AST comprehends numbers.
The border can be crossed, but we happen to know it only after some steps are taken outside the border.
Tsujishita named this property of semisets, after well-known features in non-standard analysis \cite{rob66}, as \textit{overspill principle}.

\section{A Problem of Building Common Knowledge}\label{email}

Sorites relations are useful when we try to capture certain types of structures where recurrences of the same procedures occur.
The example we focus here on is the structure of common knowledge.
The specific aspect we want to inquire about common knowledge can be summarized as the question: how many messages does it take for some informations to become common knowledge?
The answer we will provide is \textit{huge}.

To formalize structure of knowledge, we follow Aumann \cite{aumann76}'s framework for set theoretical environment.
Let $\Omega$ be a set, and $\mathscr{P}_1$ and $\mathscr{P}_2$ be its partition.
For every $\omega\in\Omega$ let
\[
\mathbf{P}_i\equiv \left\{ \langle \omega, P_i\rangle ;\ \omega \in P_i \mbox{ and } P_i\in\mathscr{P}_i\right\}.
\]
Then, for any given $\omega \in \Omega$, an event $E$ is called \textit{common knowledge at $\omega$} if $E$ includes the member of the meet $\mathscr{P}_1\wedge\mathscr{P}_2$ which contains $\omega$, that is, $\mathbf{P}_1\wedge\mathbf{P}_2(\omega)\subseteq E$ where $\mathbf{P}_1\wedge\mathbf{P}_2(\omega)\in \mathscr{P}_1 \wedge \mathscr{P}_2$. 

The condition of $E$'s being common knowledge can be put in other equivalent way.
Let $\mathcal{P}^1,\mathcal{P}^2,\cdots,\mathcal{P}^k,\mbox{ where } k\in \FN$, be a sequence of subsets of $\Omega$ satisfying $\mathcal{P}^1=\mathbf{P}_1(\omega)$ and $\mathcal{P}^n=P_i\mbox{ where }P_i\in\mathscr{P}_i\mbox{ and }\mathcal{P}^{n-1}\cap P_i\ne \emptyset$, where $i=2$ if $n$ is even and $i=1$ otherwise (the sequence is not unique since there may be multiple $P_i$s which satisfy $\mathcal{P}^{n-1}\cap P_i\ne \emptyset$).
$\omega'\in \Omega$ is \textit{reachable} from $\omega\in \Omega$ if $\omega\in\mathcal{P}^1$ and $\omega'\in\mathcal{P}^k$.
Then the statement that $E$ is common knowledge at $\omega\in\Omega$ can be stated equivalently as $E$ includes all the states $\omega'$ which are reachable from $\omega$.
\medskip

A concept of common knowledge is necessary when we try to analyze game theoretic situations.
However, it is also well-known that it differs slightly from our practices of reasoning in everyday life.
One of the problems pointed out by Rubinstein \cite{rubinstein89} is known as the electronic-mail game problem.

Rubinstein \cite{rubinstein89} examines the situation where two agents are faced with a coordination problem.
They want to coordinate their actions in accordance with the state of the nature, where two possible states, say $a$ and $b$, of nature are presupposed.
One agent, say the agent 1, knows the information on the true state while the other, say 2, doesn't.
They are located at mutually distanced place and can communicate only by electronic mails.
And e-mails are sent only when the state is $b$.
E-mails fail with a small probability.
So simply sending an e-mail doesn't guarantee that it is received by the recipient.
To confirm it, they must receive a reply.
But here, the same problem arises again.
They must receive a reply to the reply to guarantee that the first reply is received by the sender of the original message.
The process never ends and replies must go on forever before the original information becomes common knowledge.
Then Rubinstein shows us a claim: however hugely many times they send replies, the information will never become common knowledge.
\medskip

The situation can be put in Aumann's framework as follows:
let $\Omega$ be a set of triples consisting of the state of the coordination game, $a$ or $b$, and the numbers, $t$ and $t'$, of messages agent 1 and 2 sent respectively\footnote{
For example, $(a,0,0)$ represents the state in which the state of the game is $a$ and both agents send no message, since the message will be sent only when the state is $b$.

On the other hand $(b,2,1)$ represents the state in which the state of the game is $b$, the agent 1 sends 2 messages and the agent 2 sends just 1 message.
The agent 1 sends 1st message since the state of the game is $b$.
The agent 2 received the original message since $t'=1$ meaning that the agent 2 replied back.
$t=2$ also means that the agent 1 received the reply from the agent 2 since the agent 1 sends the second message.
But the agent 2 fails to receive the second replay, since the agent 2 doesn't send the third reply.
Provided that the agent 2 received the second reply, 2 must have sent back the third reply but it doesn't.
}.
Then $\Omega$ consists of $(a,0,0)$ and $(b,t,t')$ in which $t'=t$ or $t-1$.
The information partition of two agents are given as below:
\begin{align*}
\mathscr{P}_1 &\ =\ \{\{(a,0,0)\},\{(b,1,0),(b,1,1)\},\{(b,2,1),(b,2,2)\},\ldots\},\\
\mathscr{P}_2 &\ =\ \{\{(a,0,0),(b,1,0)\},\{(b,1,1),(b,2,1)\},\{(b,2,2),(b,3,2)\},\ldots\}.
\end{align*}

The problem pointed out by the e-mail game is stated in this setting as: the event $B$, which consists only of all the states whose first element is $b$, will never be common knowledge no matter how many times the message will be sent\footnote{
The validity of this claim is easily checked.
Evidently, all the elements of the information partitions of two agents are reachable from each other.
Therefore, the only event that can be common knowledge is $\Omega$.
Since $\Omega$ includes $(a,0,0)$, the event $B$ cannot be common knowledge even if hugely many messages are sent.}.
\medskip

\section{A Metric on State Spaces}

The claim we just saw may not fit our intuition.
Since, in most cases, we usually don't doubt that the event $B$ is common knowledge after sending hugely many replies each other.

This problem will be solved when we construct numbers based on AST.
It can be done by contracting the scope of reachability to \textit{finitely many} steps.
To put it in more precise manner, let us first look into the detailed structure of reachability.

To identify all reachable states, we define a function $L_i(A)$\footnote{This function can be defined by way of knowledge operator $K_i:2^\Omega\rightarrow2^\Omega$, which we examine later in the section \ref{knowledge}, as $L_i(A)=\Omega\setminus K_i(\Omega\setminus A)$.}, called link of agent $i$ from $A$, which indicates the set of all reachable states from $A\subseteq\Omega$, given agent $i$'s partition $\mathscr{P}_i$ as:
\[
L_i\left(A\right)\ \equiv\ \{ y\in\Omega\,;\,(\exists x\in{A})(y\in\mathbf{P}_i(x))\}.
\]

We also define the operator $L_{G}({A})$ that indicates the link among agents of, not only two but generally, $G$ from ${A}$ as:
\[
L_{G}\left({A}\right)\ \equiv\ \bigcup_{i\in G}L_i\left(A\right).
\]
and define $L_{G}^n(A)$ inductively as $L_{G}^1(A)=L_{G}(A)$ and $L_{G}^{n}(A)=L_{G}\left(L_{G}^{n-1}(A)\right)$.
For notational ease, we define $L_{G}^{0}\left(A\right)$ as $L_{G}^{0}\left(A\right)=A$.

Let us make sure that for any two reachable states $x$ and $y$, there exists an $n$-step link between them.
\begin{prop}
For any two mutually reachable states $x,y\in\Omega$, there exists at least one natural number $n\in N$ which satisfies $y\in L_G^n\left(\{x\}\right)$.
\end{prop}
\begin{proof}
Let $x$ and $y$ are mutually reachable states.
Since they are reachable, there exist a natural number $n\in N$ and sequence of subsets $\mathcal{P}^1,\ldots,\mathcal{P}^n\subseteq\Omega$, which satisfies $\mathcal{P}^1=\mathbf{P}_i(x)$ for some $i\in G$, $\mathcal{P}^{\ell-1}\cap\mathcal{P}^{\ell}\neq\emptyset$ and $\mathcal{P}^\ell\in \mathscr{P}_j$ for some $j\in G$ for all $\ell\in\{2,\ldots,n\}$, and $\mathcal{P}^n=\mathbf{P}_k(y)$ for some $k\in G$.
Thus, $\mathcal{P}^\ell\in L^\ell_G(\{x\})$ for all $\ell$ and $\mathcal{P}^n\subseteq L_G^n\left(\{x\}\right)$.
It implies $y\in L_G^n\left(\{x\}\right)$.
\end{proof}

$L^n_{G}(A)$ has many other tractable properties.
Let us display some of them.

\begin{lemma}
For any $A\subseteq\Omega$, $A\subseteq L_{G}(A)$.
\end{lemma}
\begin{proof}
Since $x\in \mathbf{P}_i(x)$ for any $x\in A$ and $i\in G$, $A\subseteq L_i(A)\subseteq \bigcup_{i\in G}L_i(A)$. 
\end{proof}

\begin{cor}\label{monotone}
For any $n\in N$ and $A\subseteq \Omega$, $L_{G}^{n}(A)\subseteq L_{G}^{n+1}(A)$.
\end{cor}

\begin{lemma}\label{isomorphic}
For any $A,B\subseteq\Omega$, if $A\subseteq B$ then $L_G(A)\subseteq L_{G}(B)$.
\end{lemma}
\begin{proof}
For any $y\in L_G(A)$ there exists $x\in A$ and $i\in G$ which satisfies $y\in \mathbf{P}_i(x)$.
Since $x\in A\subseteq B$, $\mathbf{P}_i(x)\subseteq L_G(B)$.
It implies $y\in L_G(B)$. 
\end{proof}

\begin{lemma}\label{exchange}
For any given $x\in\Omega$ and $n\in N$, $y\in L^n_{G}\left(\{x\}\right)$ implies $x\in L^n_{G}\left(\{y\}\right)$.
\end{lemma}
\begin{proof}
If $n=0$, it is trivially satisfied since $y=x$.

Suppose $a\in L^{n-1}_G\left(\{b\}\right)$ implies $b\in L^{n-1}_G\left(\{a\}\right)$,  and $y\in L^n_G\left(\{x\}\right)$.
Then, there exists at least one $z\in L^{n-1}_{G}\left(\{x\}\right)$ which satisfies $y\in L_G\left(\{z\}\right)$.
It implies that there exists $i\in G$ in which $y\in\mathbf{P}_i(z)$ holds, thus $z\in\mathbf{P}_i(y)$ and therefore $z\in L_G\left(\{y\}\right)$.
Applying the lemma \ref{isomorphic} $n-1$ times to $\{z\}\subseteq L_G\left(\{y\}\right)$ yields that $L_G^{n-1}\left(\{z\}\right)\subseteq L_G^n\left(\{y\}\right)$.
Then, $z\in L_G^{n-1}\left(\{x\}\right)$ implies $x\in L_G^{n-1}\left(\{z\}\right)\subseteq L_G^n\left(\{y\}\right)$.
\end{proof}

We next define a metric function between two mutually reachable states which indicates steps it takes from one state to the other.
Let $||x,\,y||=0$ iff $y\in L^0_{G}\left(\{x\}\right)$.
If $y\notin L^0_{G}\left(\{x\}\right)$, then for each $n\geq 1$
\[
||x,\,y||\,=\,
n  \quad\text{iff }\quad y\in L_{G}^{n}\left(\{x\}\right)\setminus L_{G}^{n-1}\left(\{x\}\right)
\]

It is evident that this is well-defined.
Provided contrary that $||x,\,y||=\alpha$ and $||x,\,y||=\beta$ in which $\alpha<\beta$, then both $y\in L_G^{\alpha}\left(\{x\}\right)$ and $y\notin L_G^{\beta-1}\left(\{x\}\right)$ must be satisfied at the same time.
But by corollary \ref{monotone} $L_G^{\alpha}\left(\{x\}\right)\subseteq  L_G^{\beta-1}\left(\{x\}\right)$, thus $y$ must satisfy $y\in L_G^{\beta-1}\left(\{x\}\right)$.
It is contradiction.

Let us next confirm this function actually is a metric.
\begin{prop}
A function $||\cdot,\,\cdot||$ is a metric.
\end{prop}
\begin{proof}
Since $L_G^0\left(\{x\}\right)=\{x\}$, $||x,\,y||=0$ follows only when $x=y$ is satisfied.

Symmetry follows from lemma \ref{exchange}.
The case where $||x,y||=0$ is trivial.
Suppose $||x,\,y||=n>0$.
Then $y\in L^{n}_G\left(\{x\}\right)$ and $y\notin L^{n-1}_G\left(\{x\}\right)$ since $y\in L^{n}_{G}\left(\{x\}\right)\setminus L^{n-1}_{G}\left(\{x\}\right)$.
Thus $x\in L_G^{n}(\left(\{ y\}\right)$ and $x\notin L^{n-1}_G\left(\{y\}\right)$ by lemma \ref{exchange}.
Thus, $x\in L^{n}_{G}\left(\{y\}\right)\setminus L_G^{n-1}\left(\{y\}\right)$ and $||y,\,x||=n$.

Suppose $||x,\,y||=\ell$ and $||y,\, z||=m$.
Let $||x,\,z||=k>\ell+m$, then, $z\in L_{G}^k\left(\{x\}\right)\setminus L_{G}^{k-1}\left(\{x\}\right)$.
By corollary \ref{monotone}, $z\notin L_{G}^{v}\left(\{x\}\right)$ for all $v\leq k-1$.
But since $||x,\,y||=\ell$, $y\in L_{G}^\ell\left(\{x\}\right)$.
By the assumption $||y,\,z||=m$, it follows $z\in L_{G}^m\left(\{y\}\right)$.
Therefore, it must be satisfied that $z\in L_{G}^{\ell+m}\left(\{x\}\right)$.
It is contradiction since $\ell+m\leq k-1$.
\end{proof}

\section{Subjective Reachability}

Now we can construct a sorites relation $\dotlr$ among group $G$'s perspective by making use of this metric $||\cdot,\,\cdot||$ and define altered reachability as: $x$ is \textit{subjectively reachable} from $\omega$ among the group $G$ iff $\omega\dotlr x$.
Then, an event $E$'s being common knowledge at $\omega$ among the group $G$ can be restated accordingly as: $E$ is common knowledge iff $\gal(\omega)\subseteq E$, where $\gal(\omega)$ represents a galaxy defined by $\dotlr$.

With this slight alteration, the event $B$, which cannot be common knowledge in Rubinstein \cite{rubinstein89}, turns now to be common knowledge.

\begin{prop}\label{commonknowledge}
Given the state of the game is $b$ and hugely many, say $\alpha$, messages are sent, then the event $B$ is common knowledge in AST environment.
\end{prop}
\begin{proof}
It is sufficient to consider the case where the true state is $(b,\alpha,\alpha)$, in which the state of the game is $b$ and both players sent $\alpha$ messages.
Since $||(a,0,0),(b,\alpha,\alpha)||=\alpha$ and $\dotlr$ are sorites relations, $\neg\left((a,0,0)\dotlr(b,\alpha,\alpha)\right)$ and $(a,0,0)\notin \gal\left((b,\alpha,\alpha)\right)$.
Therefore, the event $B$ contains $\gal((b,\alpha,\alpha))$ as its semiset and, thus, is common knowledge.
\end{proof}

Besides this, the event $B$ has come to have many interesting features which the original one doesn't have.
Let us review two of them.
\begin{prop}
If event $B$ is common knowledge under the true state $(b,\tau,\tau)$, then it is also common knowledge under $(b,\tau-1,\tau-1)$.
\end{prop}
\begin{proof}
Since event $B$ is common knowledge under the state $(b,\tau,\tau)$, $\neg((a,0,0)\dotlr (b,\tau,\tau))$ and thus, $\neg((a,0,0)\dotlr (b,\tau-1,\tau-1))$ by Proposition \ref{basic2}.
Consequently, $(a,0,0)\notin\gal((b,\tau-1,\tau-1))$ and $\gal((b,\tau-1,\tau-1))\subseteq B$.
Therefore, event $B$ is also common knowledge under the state $(b,\tau-1,\tau-1)$.
\end{proof}

Furthermore, the next property also holds.
\begin{prop}
If event $B$ is not common knowledge under the true state $(b,t,t)$, then it is also not common knowledge under $(b,t+1,t+1)$.
\end{prop}
\begin{proof}
Since event $B$ is not common knowledge under the state $(b,t,t)$, $(a,0,0)\dotlr (b,t,t)$ and thus, $(a,0,0)\dotlr (b,t+1,t+1)$ by Proposition \ref{basic1}.
Consequently, $(a,0,0)\in\gal((b,t+1,t+1))$ and $\gal((b,t+1,t+1))\nsubseteq B$.
Therefore, event $B$ is also not common knowledge under the state $(b,t+1,t+1)$.
\end{proof}

In short, these two properties tell us that it will not be changed whether the event $B$ is common knowledge or not, no matter how many times the number of messages both players sent changes unless it is huge.

\medskip

The result of Proposition \ref{commonknowledge} can be confirmed within the e-mail game setting.
Consequently, the other type of Nash equilibrium in this setting of AST emerges in which both agents can coordinate to gain maximum payoff when huge messages are sent.
The payoff matrix of two coordination games are given as follows:
\begin{center}
\footnotesize
\begin{tabular}{c|c|c|}
\multicolumn{3}{c}{The Game $G_a$}\\[.25cm]
& $A$ & $B$\\
\hline
$A$ & $M,\,M$ & $0,\,-L$\\
\hline
$B$ & $-L,\,0$ & $0,\,0$\\
\hline
\multicolumn{3}{c}{}\\
\multicolumn{3}{c}{state $a$}\\
\multicolumn{3}{c}{probability $1-p$}
\end{tabular}
\qquad\qquad
\begin{tabular}{c|c|c|}
\multicolumn{3}{c}{The Game $G_b$}\\[.25cm]
& $A$ & $B$\\
\hline
$A$ & $0,\,0$ & $0,\,-L$\\
\hline
$B$ & $-L,\,0$ & $M,\,M$\\
\hline
\multicolumn{3}{c}{}\\
\multicolumn{3}{c}{state $b$}\\
\multicolumn{3}{c}{probability $p$}
\end{tabular}
\end{center}
As we saw, only agent 1 knows the true state, and e-mails are sent only when the true state is $B$.
Whenever each agent's computer receives e-mails, it replies back automatically, but fails with provability of $\varepsilon>0$.
Thus, both agents cannot tell whether the opponents succeed or fail to receive the last e-mail.
However, the last proposition holds.

\begin{prop}
There is a Nash equilibrium in which agent $i$ plays $A$ when the agent sent huge number of e-mails and play $B$ otherwise.
\[
S_i(t)\ =\ \begin{cases}
A\quad\text{if }\  t\in \FN\\
B\quad\text{otherwise}
\end{cases}
\]
\end{prop}

\begin{proof}
When agent 1 receives finite number of e-mails, say $t$, then $t-1$ is also finite.
Therefore, 2 plays $A$ regardless of whether agent 2 succeed or fail to receive 1's $t$-th e-mail, and 1's utility playing $A$ is $M$ while deviating to $B$ reduces it to $-L$.
The argument is same for agent 2.

On the other hand, when agent 1 receives infinite number of e-mails, say $\tau$, then $\tau-1$ is also infinite.
Therefore, 2 plays $B$, and 1's utility playing $B$ is $M$ while deviating to $A$ reduces it to $0$.
The argument is same for agent 2.
\end{proof}

\section{Concluding Remarks}
We showed that once we adopt AST, it is possible to model the situation where informations of an event are shared only through not so much stable e-mail systems and become common knowledge, while it seems very hard, if not impossible, as long as we stick to the framework of ZF set theory.

As we saw, it is enabled by the way AST treats infinity: negation of the axiom of infinity for sets but classes.
However, AST is not the only framework which enables us to model this kind of situation.
In fact, this feature is shared by at least two other frameworks:
non-standard analysis \cite{rob66} and alternative mathematics \cite{tjst}.

For those who know non-standard analysis, it may seem quite natural and easy to translate the arguments here to terminology of non-standard analysis.
One can translate $\FN$ to $\mathbb{N}$, a set of standard natural number, which is external, and $N$ to ${}^* \mathbb{N}$, a non-standard extension of $\mathbb{N}$, which is internal\footnote{
Vop\v{e}nka \cite{vop2} called the process of getting the set of all natural number $\mathbb{N}$ of ``Cantor's Set Theory'' critically as ``complete sharpening'', and praised the attempt of non-standard analysis, as it shares the goal with AST, to extend natural numbers beyond this completely sharpened horizon of $\mathbb{N}$.
Despite the fact that they share the goal, the direction they extend the horizon differs.
While AST moves the horizon toward inside by dividing natural numbers into two subclasses, that is, $\FN$ and $N\setminus \FN$, non-standard analysis extend natural numbers beyond the horizon $\mathbb{N}$ to non-standard one, that is, ${}^*\mathbb{N}$.
As a result, every natural number of ZF set theory remains finite in non-standard analysis, while a part of them are considered to be infinite in AST.
}.

Being aware of that elements of ${}^*\mathbb{N}\setminus \mathbb{N}$ are infinite and reachable from $\emptyset$ by induction, one can model the situation also within a framework of non-standard analysis exactly the same way as we did, but, as long as we know, none has been done\footnote{Vop\v{e}nka criticize this inactive attitude of non-standard analysis writing ``the relation of Non-standard Analysis to Cantor's Set Theory is that of vassal, which is also reflected in the name ``non-standard'' natural numbers and so on''.}.
This kind of research may open the door to a whole new way to deal with the phenomena involving infinitely many repetition.

On the other hand, alternative mathematics is fairly a new framework originated by Tsujishita \cite{tjst}, aimed at founding a system of non-standard model of arithmetic directly, without recourse to infinite set theory.
To avoid unnecessary complications, numbers and its arithmetic are simply presupposed, and essential constructions are very similar to those of AST but simpler and, in some respect, more restrictive\footnote{
One example is a condition that makes up classes.
To be a class, every element of a collection of objects must have ``distinctiveness''.
Semisets and proper classes are defined not to have this distinctiveness, so that a collections of them cannot make up a class. 
%
}.

Thus, it is also easy to translate the arguments here to alternative mathematics, namely by replacing $\FN$ to $\mathbb{N}_{acc}$, a class of all accessible natural numbers, and $N$ to $\mathbb{N}$, a class of all natural numbers.

As we saw briefly, the arguments we propose here can be dealt in multiple frameworks, but only AST have been made a number of interesting researches based upon this perspective.
This is because AST is developed to offer the very tools to investigate these kinds of phenomena right from the beginning as compared with non-standard analysis, and old enough to develop various tools as compared with alternative mathematics.
So, trying to investigate the phenomena involving infinity or indefiniteness, say, endless operations or huge crowds, one can benefit a lot from their achievements and insights.
Huge amount of treasures are waiting to be excavated.

\appendix
\def\thesection{Appendix \Alph{section}}
\section{Common knowledge as a meet}
In the Section \ref{email}, we reviewed that the event is originally defined as common knowledge in Aumann \cite{aumann76} when it include the event $\mathbf{P}_1\wedge\mathbf{P}_2(\omega)$ which is an element of the meet $\mathscr{P}_1\wedge\mathscr{P}_2$.
We didn't discuss how the definition is stated in this manner in AST environment in the main body.
In this appendix, we supplement it.

The conclusion is drawn from the next proposition.

\begin{prop}
In AST setting, the meet of the group $G$, $\bigwedge_{i\in G}\mathscr{P}_i$ coincides with the class of all galaxies, or $\Omega/{\dotlr}$ (the quotient of $\Omega$ by $\dotlr$).
\[
\bigwedge_{i\in G}\mathscr{P}_i\ =\ 
\left\{\gal(\omega);\,\omega\in\Omega \right\}
\]
\end{prop}

\begin{proof}
First we show that the class of all galaxies is partition.

By definition of galaxies, two galaxies of different states $x$ and $y$ are identical or mutually disjoint.
Provided that they are not disjoint, they share at least one common element $z\in \gal(x)\cap\gal(y)$ which satisfies both $z\dotlr x$ and $z\dotlr y$.
By transitivity, $x\dotlr y$ follows.
Consequently, $\gal(x)=\gal(y)$ follows.

Since the class $\left\{\gal(\omega);\,\omega\in\Omega\right\}$ is defined for all states, the union of all the elements of the class trivially coincides with $\Omega$.

Secondly, we show that this is the smallest class which is strictly coarser than $\mathscr{P}_i$ for all $i\in G$.
Provided that there exists $\mathscr{A}$ which is finer than $\left\{\gal(\omega);\,\omega\in\Omega\right\}$, there must be an event $A\in\mathscr{A}$ which satisfies $A\subset \gal(\omega)$ for some $\omega\in\Omega$.
It means there must be at least one state $x\in \gal(\omega)\setminus A$ which satisfies $\mathcal{P}^n=\mathbf{P}_i(x)\in\mathscr{P}_i$ for some $n\in\FN$ and $i\in G$.
It means $x$ has an $n$-step link from $\omega$, which is contradiction.
\end{proof}

Clearly from the proposition, the event $E$'s being common knowledge can be stated the same way as Aumann \cite{aumann76} as $\mathbf{P}_1\wedge\mathbf{P}_2(\omega)\subseteq E$, which coincides with the definition we stated in the main body since $\mathbf{P}_1\wedge\mathbf{P}_2(\omega)=\gal(\omega)$.

It is worth mentioning that our definition of common knowledge and that of the Aumann differ only when $\gal(\omega)$ is proper semiset.
When the $\gal(\omega)$ is a set, they coincide.
It can be said alternatively that \textit{subjective} and \textit{objective} perception differ when it takes hugely many steps to get objective evidence that we both agree that we agree.

\section{Knowledge operator}\label{knowledge}
One of the major ways to analyze structure of knowledge is via knowledge operator \cite{nielsen1984common,fagin2004reasoning}.
This operator can be represented making use of the link function.
We briefly sketch how.

A Knowledge operator is defined in \cite{nielsen1984common,fagin2004reasoning} as:
\[
K_i(A) \equiv \left\{\omega\in\Omega;\, \mathbf{P}_i(\omega)\subseteq A\right\}
\]
The statement $\omega\in K_{i}(A)$ means that the agent $i$ knows the event $A$ occurs at $\omega$.

The event indicated by this operator can also be interpreted as the set of states which has no link from outside of $A$.
In other words, the agent $i$ cannot build any link from within to the outside of $A$.

\begin{prop}
The event indicated by $K_i(A)$ coincides with the event which has no link outside of $A$:
\[
K_i(A)\ =\ \overline{L_i\left(\overline{A}\right)}
\]
where $\overline{A}$ is complement of $A$, that is, $\overline{A}\equiv \Omega\setminus A$.
\end{prop}

\begin{proof}
Suppose $x\in K_i(A)$.
It is equivalent to $\mathbf{P}_i(x)\subseteq A$ and therefore to
\[
\mathbf{P}_i(x)\cap L_i\left(\overline{A}\right)\ =\ \emptyset.
\]
Finally, it is equivalent to $x\in 
\overline{L_i\left(\overline{A}\right)}$.
\end{proof}

The operator $K_{G}(A)$ is defined as:
\[
K_{G}(A)\ \equiv\ \bigcap_{i\in G}K_i\left(A\right)\ =\ \bigcap_{i\in G}\overline{L_i\left(\overline{A}\right)}
\]
Then the class of all states at which every agent knows that $A$ is common knowledge can be stated as:
\[
C_{G}(A)\ \equiv\ \left\{\omega\in\Omega;\, \left(\forall x\in \overline{A}\right)\left(\neg(x\dotlr \omega)\right)\right\}
\]
It means that $C_{G}(A)$ is subjectively unreachable from the outside of $A$.

Finally, the event $A$'s being common knowledge under the state $\omega$ comes to coincide with the fact that the state is included in $C_{G}(A)$.
\begin{prop}
$\text{gal}(\omega)\subseteq A$ iff 
$\omega\in \text{C}_{G}(A)$
\end{prop}
\begin{proof}
$\text{gal}(\omega)\subseteq A$ is equivalent to the state $\omega$'s being subjectively unreachable from the outside of $A$.
It is exactly what the statement $\omega\in \text{C}_{G}(A)$ says.
\end{proof}

\bibliographystyle{plain}
\bibliography{ref} 

\begin{thebibliography}{10}

\bibitem{aumann76}
Robert~J. Aumann.
\newblock Agreeing to disagree.
\newblock {\em Annals of Statistics}, 4(6):1236 -- 1239, 1976.

\bibitem{fagin2004reasoning}
Ronald Fagin, Joseph~Y Halpern, Yoram Moses, and Moshe Vardi.
\newblock {\em Reasoning about knowledge}.
\newblock MIT press, 2004.

\bibitem{gurican}
Jaroslav Guri\v{c}an and Pavol Zlato\v{s}.
\newblock Biequivalences and topology in the alternative set theory.
\newblock {\em Commentationes Mathematicae Universitatis Carolinae}, 26(3):525
  -- 552, 1985.

\bibitem{nielsen1984common}
Lars~Tyge Nielsen.
\newblock Common knowledge, communication, and convergence of beliefs.
\newblock {\em Mathematical Social Sciences}, 8(1):1--14, 1984.

\bibitem{rob66}
Abraham Robinson.
\newblock {\em Non-standard Analysis}.
\newblock North-Holland Publishing Co., Amsterdam, 1966.

\bibitem{rubinstein89}
Ariel Rubinstein.
\newblock {The Electronic Mail Game: Strategic Behavior under ``Almost Common
  Knowledge.''}.
\newblock {\em American Economic Review}, 79(3):385--91, June 1989.

\bibitem{tjst}
Toru Tsujishita.
\newblock Alternative mathematics without actual infinity.
\newblock {\em arXiv preprint arXiv:1204.2193v2 [math.GM]}, 2012.

\bibitem{vop2}
P.~Vop\v{e}nka.
\newblock The philosophical foundations of alternative set theory.
\newblock {\em International Journal of General Systems}, 20:115 -- 126, 1991.

\bibitem{ast}
Petr Vop\v{e}nka.
\newblock {\em Mathematics in the Alternative Set Theory}.
\newblock Teubner Verlagagesellshaft, Leipzig, 1979.

\bibitem{encycro-ast}
Petr Vop\v{e}nka and Kate\v{r}ina Trlifajov\'{a}.
\newblock Alternative set theory.
\newblock In Christodoulos~A. Floudas and Panos~M. Pardalos, editors, {\em
  Encyclopedia of Optimization}, pages 73--77. Springer, 2009.

\end{thebibliography}

\end{document}